\RequirePackage{fix-cm}
\documentclass[twocolumn]{svjour3}          
\smartqed  
\usepackage{graphicx}
\usepackage{amssymb,amsmath}
\usepackage{times}
\usepackage{color}
\usepackage{verbatim}
\usepackage{fancyhdr}
\newtheorem{assumption}{Assumption}
\begin{document}

\title{Flexible Attribute-Based Encryption Applicable to \\Secure E-Healthcare Records
}


\author{Bo Qin  \and
        Hua Deng \and
        Qianhong Wu$^*$\and
        Josep Domingo-Ferrer \and
        David Naccache \and
        Yunya Zhou 
}


\institute{Bo Qin \at
              Renmin University of China\\
              No. 59, Zhongguangun Street, Haidian District, Beijing, China
           \and
           Hua Deng \at
           School of Computer, Wuhan University, Wuhan, China
           \and
           Qianhong Wu, Corresponding author \at
           School of Electronic and Information Engineering, Beihang University\\
           XueYuan Road No.37, Haidian District, Beijing, China\\
           Tel.: 0086 10 8233 9469\\
           \email{qianhong.wu@buaa.edu.cn}
           \and
           Josep Domingo-Ferrer \at
           Universitat Rovira i Virgili, Department of Computer Engineering and Mathematics\\
           UNESCO Chair in Data Privacy, E-43007 Tarragona, Catalonia
           \and
           David Naccache \at
           \'{E}cole normale sup\'{e}rieure, D\'{e}partement d'informatique\\
           45 rue d'Ulm, F-75230, Paris Cedex 05, France
           \and
           Yunya Zhou \at
           School of Electronic and Information Engineering, Beihang University\\
           XueYuan Road No.37, Haidian District, Beijing, China
           }

\date{Received: date / Accepted: date}

\maketitle

\begin{abstract}
In e-healthcare record systems (EHRS), attribute-based encryption (ABE)
appears as a natural way to achieve fine-grained access control on health records. Some proposals exploit key-policy ABE (KP-ABE) to protect privacy in
such a way that all users are associated with specific access policies and only the ciphertexts matching the users' access policies can be decrypted.
An issue with KP-ABE is that it requires an {\em a priori} formulation
of access policies during key generation, which is not always practicable in EHRS because the policies to access health records are sometimes determined after key generation. In this paper, we  revisit KP-ABE and propose a {\em dynamic} ABE paradigm, referred to as access policy redefinable ABE (APR-ABE). To address the above issue,  APR-ABE allows users to redefine their access policies and delegate keys for the redefined ones; hence {\em a priori} precise policies are no longer mandatory. We construct an APR-ABE scheme with short ciphertexts and prove its full security in the standard model under several static assumptions.

\end{abstract}

\section{Introduction}
\label{intro}

Attribute-based encryption (ABE) provides fine-grained access control
over encrypted data by using access policies and attributes embedded in secret keys and
ciphertexts. ABE cryptostems \cite{SW05} fall into two categories:
key-policy ABE (KP-ABE) \cite{GPS+06} systems and ciphertext-policy ABE (CP-ABE)
\cite{BSW07} systems. In a CP-ABE system, the users' secret keys are associated with sets of attributes, and a sender generates a ciphertext with an access policy specifying the attributes that the decryptors must have. Alternatively, in a KP-ABE system, the users' secret keys are labeled with access policies and the sender specifies a set of attributes; only the users whose access policies match the attribute set can decrypt.

ABE requires \emph{a priori} access policies, which are not always available. This may limit its applications in practice. The following scenario
illustrates our point.

In an e-healthcare record system (EHRS),  Alice's health records are encrypted by the doctors whom she consulted before. When Alice authorizes some doctors to access her encrypted medical records, she may have no sufficient expertise to precisely determine which doctors should access the records. Instead, according to her
experience and common sense, she may specify a policy saying that the doctor
ought to be medicine professor with five-year working experience from the
hospitals she knows. After a matching doctor Bob sees Alice's medical materials, Bob finds that Alice has something wrong with her heart. Hence, a
cardiologist's advice must be sought; thus, a cardiologist (who can be a professor or not) must be allowed to see Alice's documents.

In this application, the main obstacle to apply ABE is that Alice, serving as the key generation authority, cannot generate secret keys for access policies that are {\em a priori} 
``carv\-ed in stone'', because she does not clearly know which experts are necessary for her diagnosis.

In fact, the access policy must be dynamically modified. That is, authorized users must be able to redefine their access policies and then delegate secret keys for the redefined access policies to other users. For instance, in the above motivating scenario, Alice first authorizes doctors with some general attributes to access her encrypted medical records. After the matching doctor makes a preliminary diagnosis and finds something wrong with
Alice's heart, the doctor redefines his access policy to involve some special attributes (e.g. specialty: cardiologist) and delegates to the doctor with the redefined access policy. In this way,  {\em a priori} precise access policies are not mandatory during key generation because
they can be later redefined in delegation.

There are already some ABE schemes supporting delegation. The CP-ABE schemes in \cite{BSW07,GJP+08,Wat11} allow users to delegate 
{\em more restricted secret keys, that is, keys for attribute sets that are subsets of the original ones.} 
In KP-ABE, the schemes proposed in \cite{GPS+06,LW11,BNS13,RW13} provide a delegation mechanism, but all of them require that the access policy to be delegated be more restrictive. 
This {\em limited} delegation functionality is
often insufficient: for example, in the motivating
application above 
Bob should be able to delegate to a cardiologist even if Bob
is not a cardiologist himself.
Limiting the user to delegating keys for other users associated
with more restrictive access policies
is too rigid.

The challenge of providing appropriate delegation for the applications above
has to do with the underlying secret sharing scheme. In most KP-ABE schemes (\cite{GPS+06,LW11,RW13}), secret sharing schemes are employed to share a secret in key generation and reconstruct the secret during decryption. In the key generation, each attribute involved in the access policy needs to be associated with a secret share. If there are new attributes in the target access policy to be delegated to, users cannot delegate a secret key for the access policy since they are unable to generate shares for the new attributes without knowing
the secret. This is why the above mentioned KP-ABE schemes require the 
delegated access policy to be more restrictive than the original one. This hinders
applying them for the motivating application, where the doctor with general attributes would like to delegate his access rights to a 
doctor associated with
 new special attributes.

\subsection{Our Work}
\label{contri}
We propose a dynamic primitive referred to as access policy redefinable ABE (APR-ABE). The functional goal of APR-ABE is to provide a more dynamic delegation mechanism. In an APR-ABE system, users can play the role of the key generation authority by delegating secret keys to their subordinates. The delegation does not require the redefined access policy to
be more restrictive than the one of the delegating key.

Noting that attributes are very
often hierarchically related in the real word,
we arrange the attribute universe of APR-ABE in a matrix.
For example, we can place the attribute ``Internal medicine''
at a higher level of the matrix than the attribute ``Cardiologist''.
Due to this arrangement, the notion of {\em attribute vector}
naturally comes up: an attribute vector can be generated by
picking single attributes from upper levels to lower levels.
By using attribute vectors, we can realize a delegation
that allows new attributes to be added into the original access policy
and a secret key to be delegated for the resulting policy.
This delegation is similar to the one of hierarchical identity-based
encryption (HIBE,\cite{BBG05}), but with the difference that
only delegation to the attributes consistent
with the attribute matrix is allowed.

We present an APR-ABE framework based on KP-ABE and define its
full security. In  APR-ABE, the users' secret keys are associated with
an access structure formalized by attribute vectors. Users
at higher levels can redefine their access structures and then
delegate secret keys to others in lower levels without the constraint that the redefined access structures of the delegated keys be more restrictive.
Ciphertexts are generated with sets of attribute vectors, and decryption succeeds if and only if the attribute set of a ciphertext satisfies the access structure associated with a secret key, just as in the ordinary KP-ABE.
In full security, a strong security notion in ABE systems, an adversary is allowed to access public keys, create attribute vectors and query secret keys for specified access structures. Full security states that
not even such an adversary can get any useful information about the plaintext encrypted in a ciphertext, provided that he does not have the correct decryption key.

We construct an  APR-ABE scheme by employing a linear secret sharing scheme (LSSS). An LSSS satisfies linearity, that is, new shares generated by multiplying existing shares by random factors can still reconstruct the secret. Hence, when delegating to new attributes, we create new attribute vectors by combining new attributes with existing attribute vectors and we generate shares for new attribute vectors by randomizing the shares of the existing vectors. In this way, all attribute vectors in the redefined access structure will obtain functional shares and the access structure need not to be more restrictive than the one
of the delegating key. One may attempt to trivially construct APR-ABE from HIBE by directly setting each attribute vector as the identity vector in HIBE. However, this trivial construction would suffer from collusion attacks because a coalition of users may collude to decrypt ciphertexts sent to none of them, 
even though the access structure of
none of the colluders matches the attribute sets of the concerned ciphertexts.
The proposed APR-ABE scheme withstands this kind of collusion attack by associating random values to the secret keys of users. The proposed  APR-ABE scheme has short ciphertexts and is proven to be fully secure in the standard model under several static assumptions.

APR-ABE can provide an efficient solution to the motivating application.
General attributes can be placed
in the first level and more specific, professional attributes in the next level. Alice authorizes doctors to access her medical records by specifying access policies in terms of general attributes. These authorized doctors can
redefine their access policies in terms of professional attributes
and they can delegate keys to other doctors. The matching doctors then can read Alice's records if their general and specific 
professional attributes match those
specified by the doctors who encrypted Alice's health records.


\subsection{Applying  APR-ABE to EHR Systems}
\label{application}

Our APR-ABE can be applied to EHR systems to circumvent the issue of {\em a priori} formulation of access policies. 
The APR-ABE solution relies on cleverly designed attribute hierarchies. We can arrange the attribute universe in a matrix such that general attributes like
hospital name (for example ``Hospital A'', ``Hospital B''), 
title (for example ``Professor'') 
or working years are placed in the first level, while 
specific professional attributes of doctors
(typically their medical specialty, with values like 
 ``Cardiologist'',  ``Gastroenterologist'', etc.)
 are placed in the next level. When delegating, doctors matching general attributes can redefine their access policies 
in terms of professional attributes.
We now describe how does APR-ABE work for such setting in an EHR system.

As depicted in Fig. \ref{application}, an EHR system employs a health record repository to store patients' health records. To protect privacy, all health records are encrypted by doctors who make diagnoses. Suppose that Alice's health records are encrypted with an 
attribute set $S=$\{Hospital A, Cardiologist, Professor, Working years$\geq3$\}. When Alice feels sick, she wants to authorize some doctors to read her health records. However, she may not know what exact experts are necessary for her diagnosis. Instead of generating secret keys for all doctors of Hospital A, Alice specifies an access policy $\mathbb{A}=$\{Hospital A AND Professor AND Working years$\geq5$\} and generates a secret key $SK_\mathbb{A}$ for a doctor matching this access policy. The matching doctor then makes a preliminary diagnosis on Alice's health records. Upon finding that Alice has
a heart condition,  the doctor redefines the access policy $\mathbb{A}$
to seek greater specialization, $\mathbb{A}'=$\{\{Hospital A, Cardiologist\} AND Professor AND Working years$\geq7$\} and delegates a secret key for $\mathbb{A}'$. Since the set $S$ associated with Alice's health records satisfies access structure $\mathbb{A}'$, the doctor with $\mathbb{A}'$ can decrypt and read Alice's health records. 
We note that the pair 
of attributes \{Hospital A, Cardiologist\} 
that appears in $\mathbb{A}'$ is treated as an attribute vector in our APR-ABE.
Thus in the 
redefinition of $\mathbb{A}$ as $\mathbb{A}'$, the new attribute ``Cardiologist'' can be added, that is, the delegation is {\em not} more restrictive.

\begin{figure}
  \includegraphics[width=0.50\textwidth]{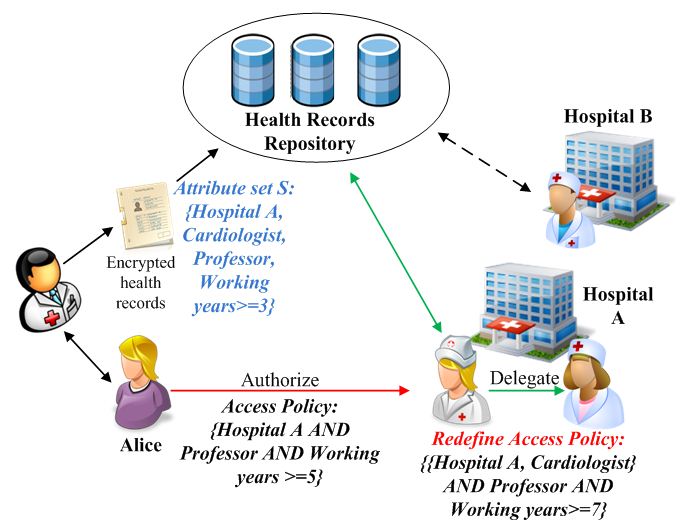}
\caption{Application to EHR systems}
\label{application}       
\end{figure}

\subsection{Paper Organization}

The rest of this paper is organized as follows. We recall the related work in Section \ref{related}. Section \ref{background} reviews the necessary background for our APR-ABE construction. We formalize the APR-ABE and define its security in Section \ref{model}. Section \ref{scheme} proposes an APR-ABE and proves its security in the standard model. Finally, we conclude the paper in Section \ref{conclusion}.

\section{Related Work}
\label{related}
%

ABE is a versatile cryptographic primitive allowing fine-grained access control over encrypted files.  ABE was introduced by Sahai and Waters
\cite{SW05}. Goyal {\it et al.} \cite{GPS+06} formulated two complementary forms of ABE, i.e., Key-Policy ABE and Ciphertext-Policy ABE, and presented the first
KP-ABE scheme. The first CP-ABE scheme was proposed by Bethencourt {\it et al.} in \cite{BSW07}, although its security proof relies on generic bilinear group
model. Ostrovsky {\it et al.} \cite{OSW07} developed a KP-ABE scheme to handle any non-monotone structure; hence, negated clauses can be included in the
policies. Waters \cite{Wat11} presented a CP-ABE construction that allows any attribute access structure to be expressed by a Linear Secret Sharing Scheme
(LSSS). Attrapadung {\it et al.} \cite{ALP11} gave a KP-ABE scheme
permitting non-monotone access structures and constant-size ciphertexts.
To reduce
decryption time, Hohenberger and Waters \cite{HW13} presented a KP-ABE with fast decryption.

The flexible encryption property of ABE made it widely adopted in e-healthcare record systems. Li {\it et al.} \cite{LYZ+13} leveraged ABE to encrypt personal health records in cloud computing and exploited multi-authority ABE to achieve a high degree of privacy of records. Yu {\it et al.} \cite{YRL11} adopted and tailored ABE for wireless sensors of e-healthcare systems. Liang {\it et al.} \cite{LBL+12} also applied ABE to secure private health records in health social networks. In their solution, users can verify
 each other's identifiers  without seeing sensitive attributes, which yields a high level of privacy. Noting that the application of KP-ABE to distributed sensors in e-healthcare systems introduces several challenges regarding attribute and user revocation, Hur \cite{Hur11} proposed an access control scheme using KP-ABE that has efficient attribute and user revocation capabilities.

In order to allow delegation of access rights to encrypted data, some ABE schemes support certain key delegation. CP-ABE \cite{BSW07,GJP+08,Wat11} allow users to delegate to attribute sets that are subsets of the original ones. Since a secret sharing scheme is used in key generation, the delegation of KP-ABE is more complicated. Goyal {\it et al.} \cite{GPS+06} adopted Lagrange interpolation to realize secret sharing and achieved a KP-ABE with selective security. This scheme supports key delegation while requiring the tree structures of delegated keys to be more restrictive than the one of the delegating key when new attributes are introduced. Lewko and Waters \cite{LW11} presented a fully secure KP-ABE  which employs a more general LSSS matrix to realize secret sharing. This KP-ABE  allows key delegation while requiring the redefined access policy to be either equivalent to the original access policy or more restrictive when new attributes need to be added. The KP-ABE in \cite{RW13} is an improvement of Lewko and Waters' KP-ABE and inherits its delegation, which
is hence limited as well. Recently, Boneh  {\it et al.} \cite{BNS13} proposed an ABE where access policies are expressed as polynomial-size arithmetic circuits. Their system supports key delegation but the size of the secret keys
increases quadratically with the number of delegations.

There are some works resolving delegation in different applications. To achieve both fine-grained access control and high performance for enterprise users, Wang {\it et al.} \cite{WLW10} proposed a solution that combines
hierarchical identity-based encryption with CP-ABE to allow a performance-expressivity tradeoff. In that scheme, various authorities rather than attributes are hierarchically organized in order to generate keys for users in their domains. Wan  {\it et al.} \cite{WLD12} extended ciphertext-policy attribute-set-based encryption with a hierarchical structure of users to achieve scalability and flexibility for access control in cloud computing systems. Li {\it et al.} \cite{LWW+11} enhanced ABE by organizing attributes in a tree-like structure to achieve delegation, which is similar to our arrangement of attributes; however,
 their delegation is still limited to
increasingly restrictive access policies. Besides, the security
of the proposed scheme is only selective.  Indeed,
all these schemes are proposed to adapt ABE for specific applications,
while our APR-ABE aims at permitting users to redefine their
access policies and delegate secret keys
in a way that does not need to be increasingly restrictive.

\section{Preliminaries}
\label{background}
In this section, we overview access structures, linear secret sharing schemes (LSSS), the composite-order bilinear group equipped with a bilinear map, and several complexity assumptions.

\subsection{Access Structures \cite{Bei96}}
\label{sec21}

\begin{definition}\label{accstr} Let $\{P_1, P_2,\cdots, P_n\}$ be a set of parties. A collection $\mathbb{A}\subseteq 2^{\{P_1,P_2,\cdots,P_n\}}$ is
monotone if for $\forall B, C$, we have that $C \in \mathbb{A}$ holds if $B \in \mathbb{A}$ and $B\subseteq C$. An access structure (respectively,
monotone access structure) is a collection (respectively, monotone collection) $\mathbb{A}$ of non-empty subsets of $\{P_1, P_2, . . . , P_n\}$, i.e.,
$\mathbb{A}\subseteq 2^{\{P_1,P_2,\cdots,P_n\}}\backslash \{\emptyset\}$. The sets in $\mathbb{A}$ are called the authorized sets, and the sets not in
$\mathbb{A}$ are called the unauthorized sets.
\end{definition}

In traditional KP-ABE, the role of the parties is played by the attributes. In our APR-ABE, the role of the parties is taken by
attribute vectors. Then an access structure is a collection of sets of attribute vectors. We restrict our attention to  monotone access structures in our APR-ABE. However we can realize general access structures by having the negation
of an attribute as a separate attribute, at the cost of doubling the number of attributes in the system.

\subsection{Linear Secret Sharing Schemes \cite{Bei96}}

\begin{definition} \label{LSSS}A secret-sharing scheme $\Pi$ over a set of parties $\mathcal{P}$ is called linear (over $\mathbb{Z}_p$) if
\begin{enumerate}
  \item The shares for each party form a vector over $\mathbb{Z}_p$.
  \item There exists a matrix $\mathbf{A}$ called the share-generating matrix for $\Pi$, where $\mathbf{A}$ has $l$ rows and $n$ columns. For all $i =
      1,\cdots, l$, the $i$-th row of $\mathbf{A}$ is labeled by a party $\rho(i)$, where $\rho$ is a function from $\{1,\cdots, l\}$ to $\mathcal{P}$.
      When we consider the column vector $\vec{s}=(s,s_2,\cdots,s_n)$, where $s \in \mathbb{Z}_p$ is the secret to be shared, and $s_2,\cdots,s_n \in
      \mathbb{Z}_p$ are randomly chosen, then $\mathbf{A}\vec{s}$ is the vector of $l$ shares of the secret s according to $\Pi$. Let $A_i$ denote the
      $i$-th row of $\mathbf{A}$, then $\lambda_i=A_i\vec{s}$ is the share belonging to party $\rho(i)$.
\end{enumerate}
\end{definition}

\medskip
\noindent\textbf{Linear Reconstruction.} \cite{Bei96} shows that every LSSS $\Pi$ enjoys the linear reconstruction property. Suppose $\Pi$
is the LSSS for access structure $\mathbb{A}$ and $S$ is an authorized set in $\mathbb{A}$, i.e., $\mathbb{A}$ contains $S$. There exist constants
$\{\omega_i\in\mathbb{Z}_p\}$ which can be found in time polynomial in the size of the share-generating matrix $\mathbf{A}$ such that if $\{\lambda_i\}$
are valid shares of $s$, then $\sum_{i\in I}\omega_i\lambda_i=s$, where $I=\{i: \rho(i)\in S\}\subseteq \{1,\cdots, l\}$.

\subsection{Composite-order Bilinear Groups}

Suppose that $\mathcal{G}$ is a group generator and $\ell$ is a
security parameter. Composite-order
bilinear groups \cite{BGN05} can be defined as
$(N=p_1p_2p_3, \mathbb{G}, \mathbb{G}_T, e)\leftarrow \mathcal{G}(1^\ell)$, where $p_1, p_2$ and $p_3$ are three distinct primes, both $\mathbb{G}$ and
$\mathbb{G}_T$ are cyclic groups of order $N$ and the group operations in both $\mathbb{G}$ and $\mathbb{G}_T$ are computable in time polynomial in
$\ell$. A map $e: \mathbb{G}\times \mathbb{G} \rightarrow \mathbb{G}_T$ is an efficiently computable map with the following properties.
\begin{enumerate}
  \item \textbf{Bilinearity}: for all $a, b\in\mathbb{Z}_N$ and $g, h\in \mathbb{G}$, $e(g^a, h^b)=e(g, h)^{ab}$.
  \item \textbf{Non-degeneracy}: $\exists g\in \mathbb{G}$ such that $e(g, g)$ has order $N$ in $\mathbb{G}_T$.
\end{enumerate}

Let $\mathbb{G}_{ij}$ denote the subgroup of order $p_ip_j$ for $i\neq j$, and $\mathbb{G}_{1}, \mathbb{G}_{2}$, $\mathbb{G}_{3}$ the subgroups of order
$p_1, p_2$, $p_3$ in $\mathbb{G}$, respectively. The \emph{orthogonality} property of $\mathbb{G}_{1}, \mathbb{G}_{2}, \mathbb{G}_{3}$ is defined as:
\begin{definition} For
all $u\in \mathbb{G}_{i}, v\in \mathbb{G}_{j}$, it holds that $e(u, v)=1$, where $i\neq j\in\{1, 2, 3\}$.
\end{definition}
The orthogonality property is essential in our
constructions and security proofs.

\subsection{Complexity Assumptions}
We now list the complexity assumptions which will be used to prove the security of our scheme. These assumptions were introduced
by \cite{LW10} to prove fully secure HIBE and
they were also employed by some ABE schemes
(e.g., \cite{LOS+10,LW11}) to attain full security.

\begin{assumption}
Let $(N=p_1p_2p_3,\mathbb{G},\mathbb{G}_T, e)\overset{R}\leftarrow \mathcal{G}(1^\ell)$. Define a distribution
$$g\overset{R}\leftarrow \mathbb{G}_{1}; X_3\overset{R}\leftarrow \mathbb{G}_{3};~ \mathcal{D}=(\mathbb{G},g,X_3); ~T_1\overset{R}\leftarrow
\mathbb{G}_{1}; ~T_2\overset{R}\leftarrow \mathbb{G}_{12}.$$ The advantage of an algorithm $\mathcal{A}$ in breaking Assumption 1 is defined as
$$\mbox{Adv1}_{\mathcal{A}}(\ell)=|\mbox{Pr}[\mathcal{A}(\mathcal{D},T_1)=1]-
\mbox{Pr}[\mathcal{A}(\mathcal{D},T_2)=1]|.$$
Assumption 1 holds if $\mbox{Adv1}_{\mathcal{A}}(\ell)$ is negligible in $\ell$ for any polynomial-time algorithm $\mathcal{A}$.
\end{assumption}

\begin{assumption}
Let $(N=p_1p_2p_3,\mathbb{G},\mathbb{G}_T, e)\overset{R}\leftarrow \mathcal{G}(1^\ell)$. Define a distribution
$$g, X_1\overset{R}\leftarrow \mathbb{G}_{1}, X_2, Y_2\overset{R}\leftarrow \mathbb{G}_{2}, X_3, Y_3\overset{R}\leftarrow \mathbb{G}_{3},$$
$$\mathcal{D}=(\mathbb{G}, g, X_1X_2, X_3, Y_2Y_3), T_1\overset{R}\leftarrow \mathbb{G}, T_2\overset{R}\leftarrow \mathbb{G}_{13}.$$
The advantage of an algorithm $\mathcal{A}$ in breaking Assumption 2 is defined as
$$\mbox{Adv2}_{\mathcal{A}}(\ell)=|\mbox{Pr}[\mathcal{A}(\mathcal{D},T_1)=1]-\mbox{Pr}[\mathcal{A}(\mathcal{D},T_2)=1]|.$$
Assumption 2 holds if any polynomial-time algorithm $\mathcal{A}$ has $\mbox{Adv2}_{\mathcal{A}}(\ell)$ negligible in $\ell$.
\end{assumption}

\begin{assumption} Let $(N=p_1p_2p_3,\mathbb{G},\mathbb{G}_T, e)\overset{R}\leftarrow \mathcal{G}(1^\ell)$. Define a distribution
$$\alpha,s\overset{R}\leftarrow \mathbb{Z}_N, g\overset{R}\leftarrow \mathbb{G}_{1},X_2, Y_2, Z_2\overset{R}\leftarrow \mathbb{G}_{2},
X_3\overset{R}\leftarrow \mathbb{G}_{3},$$ $$\mathcal{D}=(\mathbb{G}, g, g^\alpha X_2, X_3, g^sY_2, Z_2), T_1=e(g,g)^{\alpha s}, T_2\overset{R}\leftarrow
\mathbb{G}_{T}.$$ The advantage of an algorithm $\mathcal{A}$ in breaking Assumption 3 is defined as
$$\mbox{Adv3}_{\mathcal{A}}(\ell)=|\mbox{Pr}[\mathcal{A}(\mathcal{D},T_1)=1]-\mbox{Pr}[\mathcal{A}(\mathcal{D},T_2)=1]|.$$
Assumption 3 holds if any polynomial-time algorithm $\mathcal{A}$ has $\mbox{Adv3}_{\mathcal{A}}(\ell)$ negligible in $\ell$.
\end{assumption}

\section{Modeling Access Policy Redefinable Attribute-based Encryption}
\label{model}
\subsection{Notations}
We now model the APR-ABE system. First, we introduce some notations used in the description.
Observing that the hierarchical property exists among attributes in the real world, we arrange the APR-ABE attribute universe $\mathbf{U}$ in a matrix with $L$ rows and $D$ columns, that is,
$$\mathbf{U}=(u_{i,j})_{L\times D}=(U_1,\cdots,U_i,\cdots,U_L)^T,$$
where $U_i$ is the $i$-th row of $\mathbf{U}$ and contains $D$ attributes and $\mathbf{M}^T$ denotes the transposition of a matrix $\mathbf{M}$. We note that there may be some {\em empty} attributes in the matrix. In that case, we use a special character ``$\emptyset$'' to denote the empty attributes.

The attribute matrix naturally leads to the notion of attribute vector. We define an attribute vector of depth $k$ ($1\leq k\leq L$) as
$$\vec{u}=(u_1, u_2, ..., u_k),$$
where $u_i\in U_i$ for each $i$ from 1 to $k$. This means that an attribute vector of depth $k$ is formed by sampling single attributes from the first level to the $k$-th level. We note that each attribute $u_i$ actually corresponds to two subscripts $(i,j)$ denoting its position in the attribute matrix, but we drop the second subscript $j$ here to simplify notations.

We next define a set of attribute vectors. Let $S=\{\vec{u}\}$ denote a set of attribute vectors of depth $k$ and $|S|$ denote the set's cardinality.

For an attribute vector $\vec{u}'$ of depth $i$ and another attribute vector $\vec{u}$ of depth $k$, we say that $\vec{u}'$ is a {\em prefix} of $\vec{u}$ if
$\vec{u}=(\vec{u}', u_{i+1}, u_{i+2},..., u_{k}),$
where $1\leq i <k \leq L$.

As in Definition \ref{accstr}, we can define $\mathbb{A}$ as an access structure over attribute vectors of depth $k$ such that $\mathbb{A}$ is a
collection of non-empty subsets of the set of all attribute vectors of depth $k$. If for a set $S$ the condition $S\in\mathbb{A}$ holds, then we
say that $S$ is an authorized set in $\mathbb{A}$ and $S$ satisfies $\mathbb{A}$.

In an APR-ABE system, a secret key associated with an access structure $\mathbb{A}$ 
can decrypt a ciphertext generated with a set $S$ of attribute vectors if and only if $S\in\mathbb{A}$. A secret key associated with an access structure $\mathbb{A}'$ is allowed to delegate a secret key for an access structure $\mathbb{A}$ if these two access structures satisfy a natural condition. That is, each attribute vector of a set $S'\in\mathbb{A}'$ must be a prefix of an attribute vector in some set $S\in\mathbb{A}$ and all attribute vectors involved in $\mathbb{A}$ have prefixes in $\mathbb{A}'$. This guarantees that the user with access structure $\mathbb{A}'$ can use his existing shares to generate shares  for attribute vectors of authorized sets in $\mathbb{A}$. We note that in the delegation there is no requirement that the redefined access structure $\mathbb{A}$ must be
more restrictive than the original access structure $\mathbb{A}'$ when new attributes are added. This is because those new attributes can be concatenated to the end of existing attribute vectors of $\mathbb{A}'$ instead of being treated as new separate attributes that need to be assigned to new secret shares.

\subsection{System Model}
An APR-ABE system for message space $\mathbb{M}$ and access structure space $\Gamma$ consists of the following five
polynomial-time algorithms:

\begin{itemize}
  \item $(PK,MSK)\leftarrow$ \textbf{Setup}$(1^\ell)$: The setup algorithm takes no input other than the security parameter $\ell$ and outputs the
      public key $PK$ and a master secret key $MSK$.
  \item $CT \leftarrow$ \textbf{Encrypt}$(M,PK,S)$: The encryption algorithm takes as inputs a message $M$, the public key $PK$ and a set $S$ of
      attribute vectors. It outputs a ciphertext $CT$.
  \item $SK\leftarrow$ \textbf{KeyGen}$(PK, MSK, \mathbb{A})$: The key generation algorithm takes as inputs an access structure $\mathbb{A}$, the master
      secret  key $MSK$ and public key $PK$. It outputs a secret key $SK$ for the access structure $\mathbb{A}$.
  \item $SK\leftarrow$ \textbf{Delegate}$(PK, SK', \mathbb{A})$: The delegation algorithm takes as inputs a public key $PK$, a secret key $SK'$ for an
      access structure $\mathbb{A}'$ and another access structure $\mathbb{A}$. It outputs the secret key $SK$ for $\mathbb{A}$ if and only if
      $\mathbb{A}$ and $\mathbb{A}'$ satisfy the delegation condition.
  \item $M/\bot\leftarrow$ \textbf{Decrypt}$(CT,SK,PK)$: The decryption algorithm takes as inputs a ciphertext $CT$ associated with a set $S$ of
      attribute vectors, a secret key for an access structure $\mathbb{A}$, and the public key $PK$. If $S\in \mathbb{A}$, it outputs $M$;
      otherwise, it outputs a false symbol $\bot$.
\end{itemize}

The correctness property requires that for all sufficiently large $\ell\in \mathbb{N}$, all universe descriptions $\mathbf{U}$, all $(PK,MSK)\leftarrow
\mathbf{Setup}(1^\ell)$, all $\mathbb{A}\in \Gamma$, all $SK\leftarrow \mathbf{KeyGen}(PK,MSK,\mathbb{A})$ or $SK \leftarrow \mathbf{Delegate}(PK,SK',
\mathbb{A})$, all $M\in\mathbb{M}$, all $CT \leftarrow \mathbf{Encrypt}(M,PK,S)$, if $S$ satisfies $\mathbb{A}$, then
$\mathbf{Decrypt}(CT,\\SK,PK)$ outputs $M$.

\subsection{Security}
We now define the full security against chosen access structure and chosen-plaintext attacks in APR-ABE. In practice, malicious users are able to obtain the system public
key and, additionally, they may collude with other users by querying their secret keys. To capture these realistic attacks, we define an adversary  allowed to access the system public key, create attribute vectors and query secret keys for access structures he specifies. The adversary outputs two equal-length messages and a set of attribute vectors to be challenged. Then the full security states that not even such an adversary can distinguish with non-negligible
advantage the ciphertexts of the two messages
under the challenge set of attribute vectors, provided
that he has not queried the secret keys that can be used to decrypt the challenge ciphertext. Formally, the full security of APR-ABE is defined by a game played between a challenger $\mathcal{C}$ and an adversary $\mathcal{A}$ as follows.
\begin{itemize}
   \item \textbf{Setup}: The challenger $\mathcal{C}$ runs the setup algorithm and gives the public key $PK$ to $\mathcal{A}$.

   \medskip
   \item \textbf{Phase 1}: $\mathcal{A}$ sequentially makes queries $Q_1,...,Q_{q_1}$ to $\mathcal{C}$, where $Q_i$ for $1\leq i\leq q_1$
       is one of the following three types:
   \begin{itemize}
     \item {\sf Create}$(\mathbb{A})$. $\mathcal{A}$ specifies an access structure $\mathbb{A}$. In response, $\mathcal{C}$ creates a
         key for this access structure by calling the key generation algorithm, and places this key in the set $\mathcal{K}$ which is initialized to empty. He only gives $\mathcal{A}$ a reference to this key, not the key itself.
     \item {\sf Delegate}$(\mathbb{A},\mathbb{A}')$. $\mathcal{A}$ specifies a key $SK'$ associated with $\mathbb{A}'$ in the set $\mathcal{K}$ and an access structure  $\mathbb{A}$. If allowed by the delegation algorithm, $\mathcal{C}$ produces a key $SK$ for
         $\mathbb{A}$. He adds $SK$ to the set $\mathcal{K}$ and again gives $\mathcal{A}$ only a reference to it, not the actual key.
     \item {\sf Reveal}$(\mathbb{A})$. $\mathcal{A}$ specifies a key in the set $\mathcal{K}$. $\mathcal{C}$ gives this key to the
         attacker and removes it from the set $\mathcal{K}$.
   \end{itemize}

\medskip
   \item \textbf{Challenge}: $\mathcal{A}$ declares two equal-length messages $M_0$ and $M_1$ and a set $S^*$ of attribute vectors
       with an added restriction that for any revealed key $SK$ for access structure $\mathbb{A}$, $S^*\not\in \mathbb{A}$ and for any new key $SK'$ for access structure
       $\mathbb{A}'$ that can be delegated from a revealed one, $S^*\not\in \mathbb{A}'$. $\mathcal{C}$ then flips a random coin $b \in \{0,1\}$, and
       encrypts $M_b$ under $S^*$, producing $CT^*$. He gives $CT^*$ to $\mathcal{A}$.

       \medskip
   \item \textbf{Phase 2}:  $\mathcal{A}$ sequentially makes queries $Q_{q_1+1},...,Q_q$ to $\mathcal{C}$ just as in Phase 1, with the
       restriction that neither the access structure of any revealed key nor the access structure of any key that can be delegated from a revealed one
       contain $S^*$.

       \medskip
   \item \textbf{Guess}: $\mathcal{A}$ outputs a guess $b' \in \{0,1\}$.
\end{itemize}

The advantage of $\mathcal{A}$ in this game is defined as $$Adv_{\mathcal{A}}^{\mbox{APR-ABE}}=|\Pr[b=b']-1/2|.$$ We note that the model above
is for {\em chosen-plaintext } attacks and one can easily extend this model to handle {\em chosen-ciphertext} attacks by allowing decryption queries in
Phase 1 and  Phase 2.

\begin{definition}\label{KP-Security}
We say that an APR-ABE system is fully secure if all Probabilistic Polynomial-Time (PPT) attackers $\mathcal{A}$ have at most a negligible advantage in the
above game.
\end{definition}

\section{The Access Policy Redefinable Attribute-based Encryption Scheme}
\label{scheme}
In this section, we construct an APR-ABE with short ciphertexts. The proposed scheme is proven to be fully secure in the standard model.

\subsection{Basic Idea}

We first introduce the basic idea driving
the construction of the APR-ABE scheme.
We base the scheme on the KP-ABE scheme in \cite{LOS+10} and we exploit the
delegation mechanism used in several HIBE schemes (e.g., \cite{BBG05,LW10}).
The key point of this delegation mechanism is
to hash an identity vector to a group
element, which internally associates the identity vector with a ciphertext or a secret key. When introducing this mechanism into our APR-ABE, which involves multiple attribute vectors in a ciphertext or a secret key, we assign a key component to each attribute vector and randomize every key component  to resist collusion attacks.

On the other hand, LSSS have been widely used
in many ABE schemes \cite{ALP11,LOS+10,LW11,Wat11}. In our APR-ABE scheme,
an LSSS is used to generate
a share for each attribute vector of authorized sets in an access structure. The linear reconstruction property of LSSS guarantees that the shares of all attribute vectors in an authorized set can recover the secret. To realize a delegation not limited to
 more restrictive access policies, we must additionally manage to generate shares for new incoming attributes. However, without knowing the secret, delegators cannot directly generate new shares.
To overcome this problem, we concatenate the new incoming attributes to the end of existing attribute vectors to form new attribute vectors and use the existing shares to generate shares for the new attribute vectors. Specifically, to achieve the access structure control, each share of an attribute vector is blinded in the exponent of a key component. Then, to generate new shares, we lift a key component of an existing attribute vector to the power of a random exponent and define the resulting exponent as the new blinded share for the new attribute vector. Since LSSS satisfies linearity, the randomization of shares can still reconstruct the secret.

To realize the above idea, we slightly extend LSSS to handle attribute vectors. For an access structure $\mathbb{A}$, we generate an $l\times n$
share-generating matrix $\mathbf{A}$ ($l$ is the number of attribute vectors involved in $\mathbb{A}$). The inner product of the
$i$-th row vector of $\mathbf{A}$ and a vector taking the secret as the first coordinate is the share for the $i$-th row. We define an {\em injection} function $\rho$
which maps the $i$-th row of the matrix $\mathbf{A}$ to an attribute vector. Then $(\mathbf{A},\rho)$ is the LSSS for $\mathbb{A}$. The {\em injection}
function means that an attribute vector is associated with at most one row of $\mathbf{A}$.

\subsection{The Proposal}

We are now ready to describe our APR-ABE scheme, which is built from bilinear groups of a composite order $N=p_1p_2p_3$, as defined in Section 2.3. The
ciphertexts are generated in the subgroup $\mathbb{G}_1$. The keys are first generated in $\mathbb{G}_{1}$ and then randomized in $\mathbb{G}_{3}$. The
subgroup $\mathbb{G}_{2}$ is only used to implement semi-functionality in the security proofs.
\begin{itemize}
  \item $(PK,MSK)\leftarrow$ \textbf{ Setup}$(1^\ell)$: Run $(N=p_1p_2p_3,\mathbb{G},\mathbb{G}_T, \\e)\overset{R}\leftarrow \mathcal{G}(1^\ell)$. Let
      $\mathbb{G}_{i}$ denote the subgroup of order $p_i$ for $i=1,2,3$. Choose random generators $g\in \mathbb{G}_1, X_3\in \mathbb{G}_3$. Choose
      random elements $\alpha\in \mathbb{Z}_N,  v_i, h_j \in \mathbb{G}_{1}$ for all $i=1,\cdots,D$ and $j=1,\cdots,L$.
      The public key and the master secret key arei
             $$PK=\left(\mathbf{U},N,g,X_3,v_1,\cdots,v_D, h_1,\cdots,h_L, e(g,g)^\alpha\right), $$
      $MSK=\alpha.$

      \medskip
  \item $CT \leftarrow$ \textbf{Encrypt}$(M,PK,S)$: Encrypt a message $M$ under a set $S$ of attribute vectors of depth $k$. Choose a random $s\in\mathbb{Z}_N$ and compute
       $$C=Me(g,g)^{\alpha s},~ E=g^s.$$
       For each $j$ from 1 to $|S|$, choose a random element $t_j\in\mathbb{Z}_N$. Recall that for each attribute vector $\vec{u}=(u_1,u_2,...,u_k)$ of $S$, the first coordinate $u_1$ actually has two subscripts, denoted by $(1,x)$, representing that $u_1$ is the $x$-th entry of the first row in the attribute matrix. Then, choose $v_x$ corresponding to $x$ from the public key and compute
       $$C_{j,0}=v^s_x\left(h^{u_{1}}_1\cdot \cdot \cdot h^{u_{k}}_k\right)^{st_j}, C_{j,1}=g^{st_j}.$$

      Define the ciphertext (including $S$) as
       $$CT=\left(C, E, \langle C_{j,0}, C_{j,1}\rangle_{j=1,...,|S|}\right).$$

       \medskip
  \item $SK\leftarrow$ \textbf{KeyGen}$(PK,MSK,\mathbb{A})$: The algorithm generates an LSSS $(\mathbf{A},\rho)$ for $\mathbb{A}$, where
      $\mathbf{A}$ is the share-generating matrix with $l$ rows and $n$ columns, and $\rho$ maps each row of $\mathbf{A}$ to an attribute vector of
      depth $k$. Choose $n-1$ random elements $s_2,\cdots,s_{n} \in \mathbb{Z}_N$ to form a vector $$\vec{\alpha}=(\alpha, s_2,\cdots,s_{n}).$$
      For each $i$ from 1 to $l$, compute $\lambda_i=A_i\vec{\alpha}$, where $A_i$ is the $i$-th row vector of $\mathbf{A}$. Let $\vec{u}=(u_1,...,u_k)$ be the attribute vector mapped by $\rho$ from the $i$-th row. Assume that the first coordinate $u_1$ of $\vec{u}$ is the $x$-th entry of the first row in the attribute matrix and choose $v_x$ correspondingly. Then, select random elements $r_i \in \mathbb{Z}_N$ and $R_{i,0}, R_{i,1}, R_{i,2}, R_{i,k+1},\cdots, R_{i,L} \in \mathbb{G}_{3}$ to compute
      $$K_{i,0}=g^{\lambda_i}v^{r_i}_xR_{i,0}, ~~K_{i,1}=g^{r_i}R_{i,1}, $$
      $$K_{i,2}=(h^{u_1}_1\cdots h^{u_k}_k)^{r_i}R_{i,2},$$
      $$K_{i,k+1}=h^{r_i}_{k+1}R_{i,k+1}, ~\cdots, ~K_{i,L}=h^{r_i}_LR_{i,L}.$$
      Set the secret key (including
      $(\mathbf{A},\rho)$) to be
      $$SK=\langle K_{i,0}, K_{i,1}, K_{i,2}, K_{i,k+1}, \cdots, K_{i,L}\rangle_{i=1,...,l}.$$

      \medskip
  \item \noindent$SK\leftarrow$ \textbf{Delegate}$(PK,SK', \mathbb{A})$: The algorithm generates a secret key $SK$ for $\mathbb{A}$ by using
      the secret key
      $$SK'=\langle K'_{i',0}, K'_{i',1}, K'_{i',2}, K'_{i',k+1}, \cdots, K'_{i',L}\rangle_{ i'=1,...,l'}$$
      for $\mathbb{A}'$, where $\mathbb{A}'$ is an access structure over $l'$ attribute vectors of depth $k$ and $\mathbb{A}$ is an access structure
      over $l$ attribute vectors of depth $k+1$. If $\mathbb{A}$ and $\mathbb{A}'$ satisfy the delegation condition, the algorithm works as follows.

      ~~~~For each $\vec{u}$ involved in $\mathbb{A}$, find its prefix $\vec{u}'$ in $\mathbb{A}'$ such that
      $\vec{u}=(\vec{u}', u_{k+1})$. Suppose that $\vec{u}'$ is associated with the $i'$-th row of the share-generating matrix of $\mathbb{A}'$. Choose
      random elements $\gamma_i, \delta_{i}\in \mathbb{Z}_N$ and random group elements $R_{i,0}, R_{i,1}, R_{i,2}, 
R_{i,k+2}, ~\cdots, R_{i,L} \in
      \mathbb{G}_{3}$ for each $i$ from 1 to $l$. Then pick the key 
component $(K'_{i',0},$ $K'_{i',1},K'_{i',2},K'_{i',k+1},\cdots, K'_{i',L})$ of $\vec{u}'$
      from $SK'$ to compute the key component for $\vec{u}$:

       $$K_{i,0}=\left(K'_{i',0}\right)^{\gamma_{i}}v^{\delta_{i}}_{x}R_{i,0},$$
       $$K_{i,1}=\left(K'_{i',1}\right)^{\gamma_{i}}g^{\delta_{i}}R_{i,1},$$
       $$K_{i,2}=\left(K'_{i',2}\right)^{\gamma_{i}}\Big(K_{i',k+1}\Big)^{\gamma_{i}u_{k+1}}\left(h^{u_1}_1\cdots
       h^{u_{k+1}}_{k+1}\right)^{\delta_{i}}R_{i,2},$$
       $$K_{i,k+2}=\left(K'_{i',k+2}\right)^{\gamma_{i}}h^{\delta_{i}}_{k+2}R_{i,k+2}, ~~\cdots, $$
       $$K_{i,L}=\left(K'_{i',L}\right)^{\gamma_{i}}h^{\delta_{i}}_{L}R_{i,L}.$$
       This implicitly sets $r_{i}=\gamma_{i}r'_{i'}+\delta_{i}$, where $r'_{i'}$ is the random exponent used in creating the key component for
       $\vec{u}'$. The value $r_{i}$ is random since $\delta_{i}$ is picked randomly.
       Finally, output
       $$SK=\langle K_{i,0}, K_{i,1}, K_{i,2}, K_{i,k+2} \cdots, K_{i,L}\rangle_{i=1,...,l}.$$
       Note that this key is identically distributed as the one directly generated by \textbf{KeyGen}.

\medskip
  \item $M\leftarrow$ \textbf{Decrypt}$(CT,SK,PK)$: Given a ciphertext $CT=\left(C,E,\langle C_{j,0},C_{j,1}\rangle_{j=1,...,|S|}\right)$ for $S$ of attribute
      vectors of depth $k$ and a secret key
   $$SK=\langle K_{i,0}, K_{i,1}, K_{i,2}, K_{i,k+1}, \cdots, K_{i,L}\rangle_{i=1,...,l}$$ for access structure $\mathbb{A}$ over attribute vectors of
   depth $k$, if $S\in \mathbb{A}$, compute the constants $\{\omega_i\in\mathbb{Z}_N\}_{\rho(i)\in S}$ such that $$\sum_{\rho(i)\in
   S}\omega_iA_i=(1,0,\cdots,0).$$
      Let $\rho(i)$ be the $j$-th attribute vector in $S$. Compute:
      $$M'=\prod_{\rho(i)\in S}\left(\frac{e\left(E, K_{i,0}\right)\cdot
      e\left(C_{j,1},K_{i,2}\right)}{e\left(C_{j,0},K_{i,1}\right)}\right)^{\omega_i}.$$
      Output $M=C/M'$.
\end{itemize}

\begin{remark}
In the key delegation, when delegating a secret key for $\mathbb{A}$ from a secret key for $\mathbb{A}'$, an LSSS $(\mathbf{A},\rho)$ for $\mathbb{A}$ is
simultaneously generated: the share-generating matrix $\mathbf{A}$ is formed by setting the $i$-th row as $A_i=A'_{i'}\gamma_{i}$, where $A'_{i'}$ is the
$i'$-th row of the share-generating matrix of $\mathbf{A}'$; the function $\rho$ maps the $i$-th row to the attribute vector $\vec{u}$. The value
$\lambda_{i}=\gamma_i\lambda'_{i'}=\gamma_i A'_{i'}\vec{\alpha}=A_i\vec{\alpha}$ is the share for $\vec{u}$, where $\lambda'_{i'}$ is the share for
$\vec{u}'$.
\end{remark}

\noindent\textbf{Correctness.} Observe that\\
$M'=$
\begin{equation}
   \begin{aligned}
   \prod_{\rho(i)\in S}\left(\frac{e\left(g^s,g^{\lambda_i}\right)\cdot e\left(g^s,v^{r_i}_{x}\right)\cdot e\left(g^{st_j},(h^{u_1}_1\cdot \cdot \cdot
h^{u_k}_k)^{r_i}\right)}{e\left(v^s_x,g^{r_i}\right)\cdot e\left((h^{u_1}_1\cdots h^{u_k}_{k})^{st_j}, g^{r_i}\right)}\right)^{\omega_i}\nonumber
   \end{aligned}
\end{equation}
$=e\left(g,g\right)^{s\Sigma_{\rho(i)\in S}\omega_iA_i\vec{\alpha}}=e(g,g)^{s\alpha}.$

It follows that $M=C/M'$. The $\mathbb{G}_{3}$ parts are canceled out because of the orthogonality property.

\subsubsection{Computational Complexity}

We analyze the computational complexity of the 
main algorithms of the APR-ABE scheme, i.e., key generation, key delegation, encryption and decryption. The proposed scheme is built in bilinear groups 
$\mathbb{G}$ and $\mathbb{G}_T$, and most computations take place in the subgroup $\mathbb{G}_1$. Therefore we evaluate the times $t_p$ and $t_e$ 
consumed by the basic group operations, bilinear map and exponentiation in $\mathbb{G}_1$, respectively. We do not take into account the multiplication operation since it consumes negligible time compared to $t_p$ and $t_e$.

Table \ref{com-complexity} summarizes the time consumed by the main algorithms of the APR-ABE scheme. In this table, $L$ denotes the maximum depth of the system, $l$ the number of attribute vectors associated with a secret key, $l'$ the number of attribute vectors associated with a delegated key, $k$ the depth of the user delegating a key or the attribute vectors associated with a ciphertext, and $l^*$ is the number of attribute vectors of a set satisfying an access policy in the decryption. We can see that the time cost by the key generation algorithm grows linearly with the product of $L$ and $l$, but is independent of the user's depth. The time consumed by the delegation is related to the depth of the delegator and decreases as the depth grows. Encryption takes time linear in the product of the cardinality of the set $S$ and the depth of the attribute vectors in $S$. The ciphertexts of APR-ABE are short in that they are only linear in the cardinality of $S$. This makes the time consumed by decryption linear in the number of matching attribute vectors and independent of depth. This feature 
is comparable to the up-to-date KP-ABEs \cite{GPS+06,LOS+10,HW13}, which nonetheless do not allow the flexible key delegation achieved in our scheme.

\begin{table}[!t]
\renewcommand{\arraystretch}{1.5}
\caption{Computation}
\label{com-complexity}
\centering
\begin{tabular}{|c|c|}
\hline 
Algorithm & Computational Complexity \\
\hline
\bfseries Key Generation & $(L+3)\cdot l\cdot t_e$\\
\hline
\bfseries Key Delegation & $(2L-k+5)\cdot l'\cdot t_e$\\
\hline
\bfseries Encryption & $\left((k+2)|S|+2\right)\cdot t_e$ \\
\hline
\bfseries Decryption & $3l^*\cdot t_p$ \\
\hline
\end{tabular}
\end{table}

\subsubsection{Security}
The new APR-ABE scheme has full security, which means that any polynomial-time attacker cannot get useful information about the messages encrypted in ciphertexts if he does not have correct secret keys. Formally, the full security is guaranteed by Theorem \ref{KP-security}.
\begin{theorem}\label{KP-security} The Access Policy 
Redefinable Attribute-based Encryption scheme is fully secure in the standard model if Assumptions 1, 2 and 3 hold.
\end{theorem}

Our proof exploits the dual system encryption methodology \cite{Wat09}. This  approach  has been shown to be a powerful tool in proving the full security of properly designed HIBE and ABE schemes (e.g., \cite{LW10,LW11,LOS+10,Wat09}). Following this proof framework, we construct semi-functional keys and  ciphertexts for APR-ABE. A semi-functional APR-ABE key (semi-functional key for short) can be used to decrypt normal ciphertexts; and a semi-functional APR-ABE ciphertext (semi-functional ciphertext for short) can be decrypted by using normal keys. However, a semi-functional key cannot be used to decrypt a
semi-functional ciphertext. As in most proofs adopting dual system encryption, there is a subtlety that the simulator can test the nature of the challenge
key by using it to try to decrypt the challenge ciphertext. To avoid this paradox, we make sure that the decryption on input the challenge key is always
successful by cleverly setting the random values involved in the challenge key and challenge ciphertext. We also need to prove that these values are
uniformly distributed from the view of the adversary who cannot query the key able to decrypt the ciphertext.

In the following proof, we define a sequence of games arguing that an attacker cannot distinguish one game from the next. The first game is $Game_{real}$,
which denotes the real security game as defined in Definition \ref{KP-Security}. The second game is $Game_{real'}$, which is the same as $Game_{real}$
except that the attacker $\mathcal{A}$ does not ask the challenger $\mathcal{C}$ to delegate keys. The third game is $Game_0$, in which all keys are normal, but the challenge ciphertext is semi-functional. Let $q$ denote the number of key queries made by $\mathcal{A}$. For all $\nu=1,\cdots, q $, we define $Game_{\nu}$, in which the first $\nu$ keys are semi-functional and the remaining keys are normal, while the challenge ciphertext is semi-functional. Note that when $\nu=q$, in $Game_q$, all keys are semi-functional. The last game is defined as $Game_{final}$ where all keys are semi-functional and the ciphertext is a
semi-functional encryption of a random message. We will prove that these games are indistinguishable under Assumptions 1, 2 and 3.

The semi-functional ciphertexts and keys are constructed as follows.

\medskip
\noindent\textbf{Semi-functional ciphertext.} Let $g_2$ denote the generator of $\mathbb{G}_{2}$. We first invoke \textbf{Encrypt} to form a normal
ciphertext
$(\bar{C}, \bar{E}, \langle \bar{C}_{i^*,0}, \bar{C}_{i^*,1}\rangle_{i^*=1,...,|S^*|}).$
We choose a random element $c\in \mathbb{Z}_N$ and for all $i^*=1,~\cdots,~|S^*|$, select random exponents $ \varphi_{i^*}, \upsilon_{i^*} \in
\mathbb{Z}_N $. Set the semi-functional ciphertext to be
$$C=\bar{C}, ~E=\bar{E}g^c_2, ~C_{i^*,0}=\bar{C}_{i^*,0}g^{\varphi_{i^*}}_2, ~C_{i^*,1}=\bar{C}_{i^*,1}g^{\upsilon_{i^*}}_2.$$

\medskip
\noindent\textbf{Semi-functional key.} We first call algorithm \textbf{KeyGen} to form normal key
$\langle \bar{K}_{i,0}, \bar{K}_{i,1}, \bar{K}_{i,2}, \bar{K}_{i,k+1}, \cdots, \bar{K}_{i,L}\rangle_{i=1,...,l}.$
Then we choose random elements $f_i \in \mathbb{Z}_N$ for the $i$-th row of the share-generating matrix $\mathbf{A}$. We choose random elements
$\zeta_1,\zeta_2,...,\zeta_D, \eta_1,\eta_2,...,\eta_L\in\mathbb{Z}_N$ and a random vector $\vec{\vartheta}\in\mathbb{Z}^{n}_N$ . The semi-functional key
is set as:
$$K_{i,0}=\bar{K}_{i,0}g^{A_i\vec{\vartheta}+f_i\zeta_x}_2, ~K_{i,1}=\bar{K}_{i,1}g^{f_i}_2, $$
$$~K_{i,2}=\bar{K}_{i,2}g^{f_i\Sigma_{j=1}^{k}u_{j}\eta_j}_2,$$
$$K_{i,k+1}=\bar{K}_{(i,k+1)}g^{f_i\eta_{k+1}}_2, \cdots, K_{i,L}=\bar{K}_{(i,L)}g^{f_i\eta_L}_2.$$

\medskip

\begin{remark} When we use a semi-functional key to decrypt a semi-functional ciphertext, we will have an extra term
$$\prod_{\rho(i)\in S}\left(e(g_2, g_2)^{cA_i\vec{\vartheta}}e(g_2,
g_2)^{f_i\left(c\zeta_x+\upsilon_{i^*}\Sigma^k_{j=1}u_{j}\eta_j-\varphi_{i^*}\right)}\right)^{\omega_i}.$$
If $\vec{\vartheta}\cdot(1,0,\cdots,0)=0 \mod p_2$ and $c\zeta_x+\upsilon_{i^*}\Sigma^k_{j=1}u_{j}\eta_j-\varphi_{i^*}=0 \mod p_2$, then the extra term
happens to be one, which means that the decryption still works. We say that the keys satisfying this condition are \emph{nominally} semi-functional keys.
We will show that a nominally semi-functional key is identically distributed as a regular semi-functional key in the attacker's view.
\end{remark}

\begin{lemma}\label{lem1}~For any attacker $\mathcal{A}$, $$Game_{real}\mbox{Adv}_{\mathcal{A}}=Game_{real'}\mbox{Adv}_{\mathcal{A}}.$$
\end{lemma}
\begin{proof} From the construction of our APR-ABE, the keys from the key generation algorithm are identically distributed as the keys from the delegation
algorithm. Therefore, in $\mathcal{A}$'s view, there is no difference between these two kinds of games.\qed
\end{proof}
\begin{lemma}\label{lem2} If  $\mathcal{A}$ can distinguish $Game_{real'}$ from $Game_{0}$ with advantage $\epsilon$, then we can establish an algorithm
$\mathcal{B}$ to break Assumption 1 with advantage $\epsilon$.
\end{lemma}
\begin{proof} We construct an algorithm $\mathcal{B}$ to simulate $Game_{real'}$ or $Game_{0}$ to interact with $\mathcal{A}$ by using the tuple
$(g,~X_3,~T)$ of Assumption 1.

\medskip
\noindent\textbf{Setup}: Algorithm $\mathcal{B}$ selects a random $\alpha \in \mathbb{Z}_N$. For all $i=1,\cdots, D$ and $j=1,\cdots,L$, it chooses random
elements $\bar{\zeta}_i, \bar{\eta}_j \in \mathbb{Z}_N$ and computes $v_i=g^{\bar{\zeta}_i}, h_j=g^{\bar{\eta}_j}$. It provides $\mathcal{A}$ with public
key:
$$PK=\left(\mathbf{U},N,~g,~v_1,~\cdots,~v_D,~h_1,~\cdots,~h_L,~e(g,g)^\alpha\right).$$

\noindent\textbf{Key generation Phase 1}, \textbf{Phase 2}: Note that $\mathcal{B}$ knows the master key $MSK=\alpha$. Therefore, $\mathcal{B}$ can run
\textbf{KeyGen} to generate normal keys in Phase 1 and Phase 2.

\medskip
\noindent\textbf{Challenge}: $\mathcal{A}$ gives two equal-length messages $M_0$ and $M_1$, and a set $S^*=\{\vec{u}\}$ of
attribute vectors
to $\mathcal{B}$. $\mathcal{B}$ then uses the $T$ in the given tuple to form a semi-functional or normal ciphertext as follows.

$\mathcal{B}$ flips a random coin $b\in\{0,1\}$. For all $i^*=1,~\cdots,~|S^*|$, it chooses random elements $t_{i^*} \in\mathbb{Z}_N$. Finally,
it  sets the semi-functional ciphertext $CT$ to be:
$$C=M_b e(g,T)^\alpha,\quad E=T, $$
$$ C_{i^*,0}=T^{\bar{\zeta}_x}T^{(\bar{\eta}_1u_{1}+\cdots+\bar{\eta}_ku_{k})t_{i^*}}, \quad C_{i^*,1}=T^{t_{i^*}}.$$
If assuming $T=g^sg^c_2$, this implicitly sets
$$\varphi_{i^*}=c(\bar{\zeta}_x+t_{i^*}\sum^k_{j=1}u_{j}\bar{\eta}_j), ~~\upsilon_{i^*}=ct_{i^*},$$
but there is neither unwanted correlation between values $(\varphi_{i^*}\mod p_2)$ and values $(\bar{\zeta}_x, \bar{\eta}_j \mod p_2)$, nor correlation between
$(t_{i^*} \mod p_2)$ and $(\upsilon_{i^*} \mod p_2)$ by the Chinese Remainder Theorem. Thus, the $\mathbb{G}_{1}$ part of the ciphertext is unrelated to
the $\mathbb{G}_{2}$ part.

\medskip
\noindent\textbf{Guess}:  If $T\in \mathbb{G}_{12}$, $CT$ is a properly distributed semi-functional ciphertext. Hence we are in $Game_0$. If $T\in
\mathbb{G}_{1}$, by implicitly setting $T=g^s$, $CT$ is a properly distributed normal ciphertext. Hence we are in $Game_{real'}$. If $\mathcal{A}$ outputs
$b'$ such that $b'=b$, then $\mathcal{B}$ outputs 0. Therefore, with the tuple $(g,X_3, T)$, we have that the advantage of $\mathcal{B}$ in breaking
Assumption 1 is
$$\left|\mbox{Pr}[\mathcal{B}(g, X_3, T\in \mathbb{G}_{12})=0]-\mbox{Pr}[\mathcal{B}(g,X_3, T\in
\mathbb{G}_{1})=0]\right|$$
$$=|Game_{0}\mbox{Adv}_{\mathcal{A}}-Game_{real'}\mbox{Adv}_{\mathcal{A}}|=\epsilon,$$
where $Game_{0}\mbox{Adv}_{\mathcal{A}}$ is the advantage of $\mathcal{A}$ in $Game_0$ and $Game_{real'}\mbox{Adv}_{\mathcal{A}}$ is the advantage of
$\mathcal{A}$ in $Game_{real'}$.
\qed

\end{proof}

\begin{lemma}\label{lem3} If  $\mathcal{A}$ can distinguish $Game_{\nu-1}$ from $Game_{\nu}$ with advantage $\epsilon$, then we can construct an algorithm
$\mathcal{B}$ to break Assumption 2 with advantage $\epsilon$.
\end{lemma}

\begin{proof} We construct an algorithm $\mathcal{B}$ to simulate $Game_{\nu-1}$ or $Game_{\nu}$ to interact with $\mathcal{A}$ by using the tuple
$(g, X_1X_2, X_3,\\ Y_2Y_3, T)$ of Assumption 2.

\medskip
\noindent\textbf{Setup}: The public key $PK$ generated by $\mathcal{B}$ is the same as that in Lemma 2. Algorithm $\mathcal{B}$ gives $PK$ to
$\mathcal{A}$.

\medskip
\noindent\textbf{Challenge}: For convenience, we bring the Challenge phase before Phase1. This will not affect the security proof. When $\mathcal{A}$
queries the challenge ciphertext with two equal-size messages $M_0, M_1$ and a set $S^*$ of attribute vectors, $\mathcal{B}$ flips a random coin
$b\in\{0,1\}$ and randomly chooses $t_{i^*}\in\mathbb{Z}_N$ for all $i^*=1,\cdots, |S^*|$. It sets the ciphertext to be
        $$C=M_b e(g,X_1X_2)^{\alpha}, E=X_1X_2,$$
        $$C_{i^*,0}=(X_1X_2)^{\bar{\zeta}_x}(X_1X_2)^{(\bar{\eta}_1u_{1}+\cdots+\bar{\eta}_ku_{k})t_{i^*}}, $$
        $$C_{i^*,1}=(X_1X_2)^{t_{i^*}}.$$
By assuming $X_1X_2=g^sg^c_2$, this implicitly sets
$$\varphi_{i^*}=c(\bar{\zeta}_x+t_{i^*}\sum^k_{j=1}u_{j}\bar{\eta}_j), ~\upsilon_{i^*}=ct_{i^*}.$$
Again there is no correlation between values $(\varphi_{i^*} \mod  p_2)$ and values $(\bar{\zeta}_x, \bar{\eta}_j \mod  p_2)$, nor is there any correlation between $(t_{i^*}
\mod  p_2)$ and $(\upsilon_{i^*} \mod  p_2)$ by the Chinese Remainder Theorem. Thus the $\mathbb{G}_{1}$ part is unrelated to the $\mathbb{G}_{2}$ part of
this ciphertext. Therefore, this ciphertext is a well distributed semi-functional ciphertext.

\medskip
\noindent\textbf{Key generation Phase 1}, \textbf{Phase 2}: For the first $\nu-1$ key queries, $\mathcal{B}$ simulates the semi-functional keys. For a
queried $\mathbb{A}$, it first calls the key generation algorithm to generate an LSSS $(\mathbf{A},\rho)$ and a normal key $\langle \bar{K}_{i,0},
\bar{K}_{(i,1)}, \bar{K}_{i,2}, \bar{K}_{i,k+1},\\ \cdots, \bar{K}_{i,L}\rangle_{\forall i\in[l]}$ for this LSSS. Then, for each $i$ from 1 to $l$, $\mathcal{B}$ picks a random element $\bar{f}_i\in\mathbb{Z}_N$. $\mathcal{B}$ also chooses random elements
$\zeta_1,...,\zeta_D,\eta_1,...,\eta_L\in\mathbb{Z}_N$. Finally $\mathcal{B}$ chooses a random vector $\bar{\vec{\vartheta}}\in\mathbb{Z}^{n}_N$ and computes the
secret key:
$$K_{i,0}=\bar{K}_{i,0}(Y_2Y_3)^{A_i\bar{\vec{\vartheta}}+\bar{f}_i\zeta_x},~K_{i,1}
=\bar{K}_{i,1}(Y_2Y_3)^{\bar{f}_i},$$
$$K_{i,2}=\bar{K}_{i,2}(Y_2Y_3)^{\bar{f}_i\sum^k_{j=1}u_{j}\eta_j},$$
$$K_{i,k+1}=\bar{K}_{i,k+1}(Y_2Y_3)^{\bar{f}_i\eta_{k+1}},~\cdots,$$
$$K_{i,L}=\bar{K}_{i,L}(Y_2Y_3)^{\bar{f}_i\eta_{L}}.$$
If we assume $Y_2=g^a_2$ for some $a$, this implicitly sets $\vec{\vartheta}=a\bar{\vec{\vartheta}}, f_i=a\bar{f}_i.$ Thus this key is identically
distributed as the semi-functional key.

For the rest of key queries but the $\nu$-th one, $\mathcal{B}$ simulates normal keys. Because $\mathcal{B}$ knows the master key $MSK=\alpha$, it can easily
create the normal keys by running the key generation algorithm.

To respond to the $\nu$-th key query on an access structure $\mathbb{A}$, algorithm $\mathcal{B}$ will either simulate a normal key or a semi-functional key
depending on $T$. Algorithm $\mathcal{B}$ generates an LSSS for $\mathbb{A}$ to prepare for key generation. It creates a vector $\bar{\vec{\alpha}}$ with
the first coordinate equal to $\alpha$ and the remaining $n-1$ coordinates picked randomly in $\mathbb{Z}_N$. $\mathcal{B}$ also creates a vector
$\bar{\vec{\vartheta}}$ with the first coordinate equal to 0 and the remaining $n-1$ coordinates picked randomly in $\mathbb{Z}_N$. For each row $A_i$ of
$\mathbf{A}$, $\mathcal{B}$ chooses random elements
$\bar{r}_i \in \mathbb{Z}_N$ and  $R_{i,0}, R_{i,1}, R_{i,2}, R_{i,k+1}, \cdots, R_{i,L}\in \mathbb{G}_{3}.$
Then $\mathcal{B}$ computes:
$$K_{i,0}=g^{A_i\bar{\vec{\alpha}}}T^{A_i\bar{\vec{\vartheta}}}T^{\bar{r}_i\bar{\zeta}_x}R_{i,0},~K_{i,1}
=T^{\bar{r}_i}R_{i,1},$$
$$K_{i,2}=T^{(\bar{\eta}_1u_{1}+\cdots+\bar{\eta}_ku_{k})\bar{r}_i}R_{i,2},$$
$$K_{i,k+1}=T^{\bar{r}_i\bar{\eta}_{k+1}}R_{i,k+1},~\cdots,~K_{i,L}=T^{\bar{r}_i\bar{\eta}_L}R_{i,L}.$$
If $T\in \mathbb{G}_{13}$, by assuming $T=g^{c_1}g^{c_3}_3$, this implicitly sets
$r_i=c_1\bar{r}_i$ and $\vec{\alpha}=\bar{\vec{\alpha}}+c_1\bar{\vec{\vartheta}}.$
Thus this key is identically distributed as the normal key. If $T\in \mathbb{G}$, by assuming $T=g^{c_1}g^{c_2}_2g^{c_3}_3$, this implicitly sets:
$$r_i=c_1\bar{r}_i, f_i=c_2\bar{r}_i,\vec{\vartheta}=c_2\bar{\vec{\vartheta}}, \vec{\alpha}=\bar{\vec{\alpha}}+c_1\bar{\vec{\vartheta}},$$
and
$\zeta_1=\bar{\zeta}_1,...,\zeta_D=\bar{\zeta}_D, \eta_1=\bar{\eta}_1,...,\eta_L=\bar{\eta}_L.$
Since $r_i$ are created by $\bar{r}_i$ in $\mathbb{G}_{1}$ and $f_i$ are created by $\bar{r}_i$ in $\mathbb{G}_{2}$, there is no unwanted correlation
between the $\mathbb{G}_{1}$ part and the $\mathbb{G}_{2}$ part by the Chinese Remainder Theorem. Similarly, the fact
$\zeta_1=\bar{\zeta}_1,...,\zeta_D=\bar{\zeta}_D, \eta_1=\bar{\eta}_1,...,\eta_L=\bar{\eta}_L$ will not result in unwanted correlation between the
$\mathbb{G}_{1}$  and the $\mathbb{G}_{2}$  of this key.

When the simulator $\mathcal{B}$ uses the $\nu$-th key to decrypt the semi-functional ciphertext  to test whether the key is normal or
semi-functional, it will obtain
\begin{equation}\label{KP-formula}
        \prod_{\rho(i)\in S^*}\left(e(g_2, g_2)^{cA_i\vec{\vartheta}}e(g_2,
        g_2)^{f_i\left(c\zeta_x+\upsilon_{i^*}\Sigma^k_{j=1}u_{j}\eta_j-\varphi_{i^*}\right)}\right)^{\omega_i} \nonumber
\end{equation}
$=1.$
This is because from the simulation of semi-functional ciphertext we have that $$\varphi_{i^*}=c(\bar{\zeta}_x+t_{i^*}\sum^k_{j=1}u_{j}\bar{\eta}_j),
~\upsilon_{i^*}=ct_{i^*}$$ and from the simulation of the $\nu$-th key, we have that $$\zeta_1=\bar{\zeta}_1,...,\zeta_D=\bar{\zeta}_D,
\eta_1=\bar{\eta}_1,...,\eta_L=\bar{\eta}_L.$$  Moreover, since the inner product
$$\vec{\vartheta}\cdot(1,0,\cdots,0) = c\bar{\vec{\vartheta}}\cdot(1,0,\cdots,0) =0, $$
$\sum_{\rho(i)\in S^*}\omega_iA_i\vec{\vartheta}=0.$
Thus, when $\mathcal{B}$ uses the ${\nu}$-th key to decrypt the semi-functional ciphertext, the decryption will still work and the ${\nu}$-th key is
nominally semi-functional. Now, we argue that the nominally semi-functional key is identically distributed as a semi-functional key in $\mathcal{A}$'s
view. That is, if $\mathcal{A}$ is prevented from asking the ${\nu}$-th key that can decrypt the challenge ciphertext, the fact that $\vartheta_1=0$
($\vartheta_1$ is set as the first coordinate of $\vec{\vartheta}$) should be information-theoretically hidden in $\mathcal{A}$'s view.

Because the ${\nu}$-th key cannot decrypt the challenge ciphertext, the vector $(1,0,\cdots,0)$ is linearly independent of ${\bf R}_{S^*}$, which is a
submatrix of $\mathbf{A}$ and contains only those rows that satisfy $\rho(i)\in S^*$. From the basics of linear algebra and similarly to Proposition 11 in
\cite{LOS+10}, we have the following proposition:

\begin{proposition}~A vector $\vec{v}$ is linearly independent of a set of vectors represented by a matrix $\mathbf{M}$ if and only if there exists a
vector $\vec{w}$ such that $\mathbf{M}\vec{w}=\vec{0}$ while $\vec{v}\cdot\vec{w}=1.$
\end{proposition}

Since $(1,0,\cdots,0)$ is linearly independent of $\mathbf{R}_{S^*}$, a vector $\vec{w}$ must exist such that for each row $A_i \in
\mathbf{R}_{S^*}$, it holds that
$A_i\cdot\vec{w}=0,\vec{w}\cdot(1,0,\cdots,0)=1.$
Then for the fixed vector $\vec{w}$, we can write
$$A_i\cdot\vec{\vartheta}=A_i\cdot(z\vec{w}+\vec{r}),$$
where $\vec{r}$ is uniformly distributed and reveals no information about $z$. We note that $\vec{\vartheta}\cdot (1,0,...,0)$ can not be determined from
$\vec{r}$ alone, some information about $z$ is also needed.
If $\rho(i)\in S^*$, then $A_i\cdot\vec{\vartheta}=A_i\cdot\vec{r}$. Thus, no information about $z$ is revealed and $A_i\cdot\vec{\vartheta}$ is hidden.
If $\rho(i)\not\in S^*$, then $A_i\cdot\vec{\vartheta}+f_i\zeta_x=A_i\cdot(z\vec{w}+\vec{r})+f_i\zeta_x.$
This equation introduces a random element $f_i\zeta_x$, where $f_i$ is random and appears only once because $\rho$ is injective. Hence if not
all of the $f_i$ values are congruent to 0 modulo $p_2$, no information about $z$ is revealed. The probability that all $f_i$'s are 0 modulo $p_2$ is negligible.
Therefore, the value being shared in $\mathbb{G}_{2}$ is information-theoretically hidden in $\mathcal{A}$'s view with probability close to 1. Hence,
$\mathcal{B}$ simulates the semi-functional keys with a probability close to 1.

\medskip
\noindent\textbf{Guess}: If $T\in \mathbb{G}_{13}$, we are in $Game_{\nu-1}$. If $T\in \mathbb{G}$, we are in $Game_{\nu}$. If $\mathcal{A}$ outputs
$b'=b$, $\mathcal{B}$ outputs 0. Then, with the input tuple $(g, X_1X_2, X_3, Y_2Y_3, T)$, the advantage of $\mathcal{B}$ in breaking Assumption 2 is:
$$|\mbox{Pr}[\mathcal{B}(g, X_1X_2, X_3, Y_2Y_3, T\in \mathbb{G}_{13})=0]-$$$$\mbox{Pr}[\mathcal{B}(g, X_1X_2, X_3, Y_2Y_3, T\in \mathbb{G})=0]|$$
$$=|Game_{\nu-1}\mbox{Adv}_{\mathcal{A}}-
Game_{\nu}\mbox{Adv}_{\mathcal{A}}|=\epsilon,$$

\noindent where $Game_{\nu-1}\mbox{Adv}_{\mathcal{A}}$ is the advantage of $\mathcal{A}$ in $Game_{\nu-1}$ and $Game_{\nu}\mbox{Adv}_{\mathcal{A}}$ is the
advantage of $\mathcal{A}$ in $Game_{\nu}$.
\qed
\end{proof}

\begin{lemma}\label{lem4}~If  $\mathcal{A}$ can distinguish $Game_{q}$ from $Game_{final}$ with advantage $\epsilon$, then we can construct an algorithm
$\mathcal{B}$ that contradicts Assumption 3 with advantage $\epsilon$.
\end{lemma}

\begin{proof} We construct $\mathcal{B}$ to simulate $Game_{q}$ or $Game_{final}$ to interact with $\mathcal{A}$ by using the tuple $(g, g^\alpha
X_2, X_3, g^sY_2, Z_2, T)$ of Assumption 3.

\medskip
\noindent\textbf{Setup}: For all $i=1,~\cdots, D$ and all $j=1,~\cdots, L$, $\mathcal{B}$ chooses random exponents $\bar{\zeta}_i, \bar{\eta}_j \in
\mathbb{Z}_N$ and computes $v_i=g^{\bar{\zeta}_i}, h_j=g^{\bar{\eta}_j}.$
Then it sets
$$PK=(\mathbf{U},N, g, v_1, \cdots, v_D, h_1, \cdots, h_L, e(g, g^\alpha X_2))$$ and gives $PK$ to $\mathcal{A}$. We note that $\mathcal{B}$ does not know
the secret $\alpha$.

\medskip
\noindent\textbf{Key generation Phase 1}, \textbf{Phase 2}: To simulate the semi-functional keys for $\mathbb{A}$, $\mathcal{B}$ first generates an LSSS
$(\mathbf{A},\rho)$ for $\mathbb{A}$. It then selects two vectors: $\vec{\phi}$, which has the first coordinate set to 1 and the remaining $n-1$
coordinates randomly chosen in $\mathbb{Z}_N$, and $\vec{\psi}$, which has the first coordinate set to 0 and the remaining $n-1$ coordinates randomly
chosen in
$\mathbb{Z}_N$. We note that this implicitly sets $\vec{\alpha}=\alpha\vec{\phi}+\vec{\psi}.$

For the $i$-th row $A_i$ of $\mathbf{A}$, algorithm $\mathcal{B}$ chooses random elements $r_i, \bar{f}_i\in \mathbb{Z}_N;  R_{i,0}, R_{i,1}, R_{i,2},
R_{i,k+1}, \cdots, R_{i,L}\in \mathbb{G}_{3}.$
$\mathcal{B}$ randomly chooses $\zeta_1,...,\zeta_D, \eta_1,...,\eta_L \in \mathbb{Z}_N$ and computes the key as follows:
\begin{equation}
       \begin{aligned}
       K_{i,0}&=g^{A_i\vec{\psi}}\left(g^\alpha X_2\right)^{A_i\vec{\phi}}v^{r_i}_{x}Z^{\bar{f}_i\zeta_x}_2R_{i,0},\\
       K_{i,1}&=g^{r_i}Z^{\bar{f}_i}_2R_{i,1},\\
       K_{i,2}&=\left(h^{u_{1}}_1\cdots h^{u_{k}}_k\right)^{r_i}Z^{\bar{f}_i\sum^k_{j=1}u_{j}\eta_j}_2R_{i,2},\\
       K_{i,k+1}&=h^{r_i}_{k+1}Z^{\bar{f}_i\eta_{k+1}}_2R_{i,k+1},\\
       &\vdots\\
       K_{i,L}&=h^{r_i}_{L}Z^{\bar{f}_i\eta_L}_2R_{i,L}.\nonumber
       \end{aligned}
       \end{equation}
By assuming $X_2=g^{c_2}_2, Z_2=g^{d_2}_2$, this implicitly sets $\vec{\vartheta}=c_2\vec{\phi}, f_i=d_2\bar{f}_i.$
We also note that the values being shared in $\mathbb{G}_{2}$ are properly randomized by $f_i$. Therefore, this key is identically distributed as the
semi-functional key in $\mathcal{A}$'s view.

\medskip
\noindent\textbf{Challenge}: When $\mathcal{B}$ is given two equal-length messages $M_0$ and $M_1$ and a set $S^*$ of attribute vectors, $\mathcal{B}$ flips a
random coin $b\in\{0,1\}$ and chooses $t_{i^*}\in\mathbb{Z}_N$ for all ${i^*}=1,\cdots, |S^*|$. Then it sets the ciphertext to be:
$$C=M_b T, ~E=g^sY_2, $$
$$C_{i^*,0}=(g^sY_2)^{\bar{\zeta}_x}(g^sY_2)^{t_{i^*}(\bar{\eta}_1u_{1}+\cdots+\bar{\eta}_ku_{k})}, $$
$$C_{i^*,1}=(g^sY_2)^{t_{i^*}}.$$
Assuming $Y_2=g^c_2$, this implicitly sets $$\varphi_{i^*}=c(\bar{\zeta}_x+t_{i^*}\sum^k_{j=1}u_{j}\bar{\eta}_j)$$ and  $\upsilon_{i^*}=ct_{i^*},$
but again there is neither correlation between $(\varphi_{i^*} \mod  p_2)$ and $(\bar{\zeta}_x, \bar{\eta}_j \mod  p_2)$, nor correlation between
$(t_{i^*} \mod p_2)$ and $(\upsilon_{i^*} \mod  p_2)$ by the Chinese Remainder Theorem.

If $T=e(g,g)^{\alpha}$, then this ciphertext is the semi-functional ciphertext of message $M_b$. If $T$ is a random element in $\mathbb{G}_T$, this
ciphertext is a semi-functional encryption of a random message.

\medskip
\noindent\textbf{Guess}: If $T=e(g,g)^{\alpha}$, we are in $Game_{q}$. If $T$ is a random element in $\mathbb{G}_T$, we are in $Game_{final}$.
$\mathcal{B}$ outputs 0 when $\mathcal{A}$ outputs $b'=b$. Given the tuple $(g, g^\alpha X_2, X_3, g^sY_2, Z_2, T)$, the advantage of
$\mathcal{B}$ in breaking Assumption 3 is:
$$|\mbox{Pr}[\mathcal{B}(g, g^\alpha X_2, X_3, g^sY_2, Z_2, T=e(g,g)^\alpha)=0]-$$
$$\mbox{Pr}[\mathcal{B}(g, g^\alpha X_2, X_3, g^sY_2, Z_2,
T\overset{R}\longleftarrow \mathbb{G}_T)=0]|$$
$$=|Game_{q}\mbox{Adv}_{\mathcal{A}}-
Game_{final}\mbox{Adv}_{\mathcal{A}}|=\epsilon,$$
where $Game_{q}\mbox{Adv}_{\mathcal{A}}$ is the advantage of $\mathcal{A}$ in $Game_q$ and $Game_{final}\mbox{Adv}_{\mathcal{A}}$ is the advantage of
$\mathcal{A}$ in $Game_{final}$.
\qed
\end{proof}
From all the lemmas proven above, the proof of Theorem 1 follows:
\begin{proof}
In $Game_{final}$, the ciphertext completely hides the bit $b$, so the advantage of $\mathcal{A}$ in this game is negligible. 
Through Lemmas~\ref{lem1}, \ref{lem2}, \ref{lem3} and~\ref{lem4},
we have shown that the real security game $Game_{real}$ is indistinguishable from $Game_{final}$. Therefore, the advantage of
$\mathcal{A}$ in $Game_{real}$ is negligible. Hence, there is no polynomial-time adversary with a non-negligible advantage in breaking our
APR-ABE system. This completes the proof of Theorem 1. \qed
\end{proof}

\section{Conclusion}
\label{conclusion}
We revisited KP-ABE and proposed a dynamic ABE  referred to as APR-ABE. APR-ABE distinguishes itself from other KP-ABE schemes by providing a delegation mechanism that allows a user to redefine the access policy and delegate a secret key without making the redefined access policy more restrictive. This feature renders APR-ABE especially suitable to e-healthcare record systems where {\em a priori} specification
of access policies for secret keys is too rigid or simply not practical.
We constructed an APR-ABE scheme with short ciphertexts and proved its full security in the standard model under several non-interactive assumptions.

\section*{Acknowledgements and disclaimer}
We thank the anonymous reviewers for their valuable suggestions.
The following funding sources are gratefully acknowledged:
Natural Science Foundation of China (projects 61370190, 61173154, 61272501, 61402029, 61202465 and 61472429), China National Key Basic Research Program (973 program, project 2012CB315905), Beijing Natural Science Foundation 
(project 4132056), Fundamental Research Funds for the Central Universities of 
China, Research Funds of Renmin University (No. 14XNLF02), European 
Commission (projects FP7 ``DwB'', FP7 ``Inter-Trust'' and H2020 ``CLARUS''), 
Spanish Govt. (project TIN2011-27076-C03-01), Govt. of Catalonia 
(ICREA Acad\`emia Prize to the fourth author). The fourth author
leads the UNESCO Chair in Data Privacy, but the views in this paper
do not commit UNESCO.  


\begin{figure}
\label{fig:1}       
\end{figure}
%
\begin{figure*}
\label{fig:2}       
\end{figure*}
%





\end{document}